\documentclass[aps,pra,reprint,floatfix,amsmath,amssymb,longbibliography]{revtex4-2}
\usepackage[T1]{fontenc}
\usepackage{lmodern,amsthm,bm}
\usepackage{graphicx,booktabs,array,tabularx}
\usepackage{xcolor,microtype}
\usepackage{hyperref}
 \usepackage{placeins}
\graphicspath{{figures/}}
\newtheorem{proposition}{Proposition}

\definecolor{RevisionColor}{RGB}{170,0,0}
\newif\ifshowchanges
\showchangesfalse

\newcommand{\J}{\mathcal J}

\newcommand{\GraphCount}{21,766}
\newcommand{\PairMatches}{152}
\newcommand{\PairTotal}{270}
\newcommand{\PairPct}{56.3}
\newcommand{\RightPairQ}{1.000}
\newcommand{\WrongPairQ}{0.080}
\newcommand{\WrongPairOracleQ}{1.000}
\newcommand{\WrongPairGain}{0.735}

\begin{document}
\title{Gap Sensitivity and Bottleneck Migration in Sparse Constraint-Preserving Quantum Annealing}
\author{Aruna Gupta}
\email{arunagupta169@gmail.com}
\author{S. R. Hassan}
\email{shassan@imsc.res.in}
\affiliation{Department of Theoretical Physics, The Institute of Mathematical Sciences,
A CI of Homi Bhabha National Institute, CIT Campus, Taramani,
Chennai 600113, Tamil Nadu, India}
\begin{abstract}
Constraint-preserving exchange drivers allow many coupling graphs that preserve the same feasible sector, but exhibit different minimum gaps along an annealing path. We derive an exchange-resolved sensitivity of the normalised minimum gap for fixed-cardinality Laplacian drivers sharing a Dicke ground state, and apply it to identify locally favourable modifications of the driver. Numerical experiments indicate this sensitivity reliably predicts beneficial infinitesimal changes, but its reliability decreases when the same direction is extended to a finite edge replacement. The gap landscape tracking reveals one mechanism of this loss of reliability. Improving the active bottleneck can expose a competing minimum, with a more extensive reorganisation of the low-energy subspace. We call this the ``repair-and-expose'' mechanism
and observe it in both dense quadratic and sparse frustrated objective
families. This motivates relinearizing the response when competing
minima become co-active. In the tested oracle setting, a tied-minimum
max--min update is substantially more reliable than continuing the
original direction or optimizing only the incoming minimum. Supporting
finite-size search and selector tests show that this local information
can remain useful beyond the infinitesimal limit, while its practical
extraction depends on the driver geometry and identification of the
relevant low-energy competitors. These results connect local
exchange-resolved gap sensitivity to bottleneck migration and
conditional driver redirection.
\end{abstract}
\maketitle
\section{Introduction}
\label{sec:intro}

Quantum annealing encodes an objective in a problem Hamiltonian and
interpolates from a driver whose ground state can be initialized
\cite{farhi2000,albash2018}. The minimum gap between the relevant
low-energy eigenstates is an important property of this path, although
it does not by itself determine finite-time success. Here we consider
the narrower problem of designing a constraint-preserving driver to
improve the minimum gap of a specified annealing path.

Hard constraints can be enforced dynamically by choosing a driver that
commutes with the corresponding constraint operator
\cite{hen2016a,hen2016b,leipold2022}. For fixed Hamming weight,
two-body exchange interactions move excitations while preserving their
total number. Constraint preservation, however, does not determine
which pairs should be coupled. Many connected exchange graphs preserve
the same feasible sector and share the same initial ground state, while
their different kinetic connections can produce different minimum
gaps. Ring and complete-graph $XY$ mixers provide established reference
constructions \cite{cook2020,wang2020}, while sparse native coupler
structures motivate treating the number of exchanges as a resource
\cite{dwavetopology}. This leads to our central question:
how does changing a physical exchange affect the actual
minimum-gap bottleneck of the annealing path?

Constraint-preserving mixers
\cite{hen2016a,hen2016b,leipold2022,hadfield2019,wang2020},
minimum-gap path optimization \cite{zeng2016}, Fiedler-based graph
design \cite{ghosh2006}, and driver-based anticrossing mitigation
\cite{dickson2012,lanting2017} provide the relevant background.
Modified drivers and catalysts can also create additional gap minima and reorganise low-energy states \cite{choi2021,feinstein2024}.
Our focus is distinct: we decompose the normalised minimum-gap response into different physical exchange channels and ask how this local response changes when the driver is changed by a finite amount.

For Laplacian exchange drivers the local response has a direct physical meaning. Every physical exchange induces a family of swaps between feasible configurations, and its Dirichlet form quantifies how the wavefunction varies over those swaps. Whether that exchange locally raises the active minimum gap is determined by the difference of this quantity in the first excited and ground states, along with the driver-normalization correction. This gives a
directional principle for modifying the driver, but only locally:
a finite change can relieve the original bottleneck until another
minimum becomes limiting.

This distinction between \emph{local repair} and \emph{bottleneck
migration} is the central narrative of this work. Using quantum-state
oracle information, we first verify that the exchange-resolved
sensitivity predicts beneficial infinitesimal changes. We then follow
the same directions toward finite edge replacements and track the
coexisting gap minima and low-energy subspaces. The calculations show
that loss of finite-step reliability is associated with migration of
the active bottleneck and stronger low-energy reorganization. At
resolved bottleneck exchanges, recomputing the local direction while
protecting all active minima provides a substantially more reliable
local update. The same qualitative repair-and-expose mechanism is
observed in an independent sparse frustrated objective family.

Finally, we separate this mechanism study from the question of practical
driver selection. Exhaustive references and evaluation-cap-matched controls
test the usefulness of oracle-informed ranking for finite discrete search,
while spectral and problem-informed selectors assess approximations that do
not use the current quantum eigenstates to choose the practical proposal.
The retrospective oracle-pair construction is retained only as a diagnostic
reference. These calculations are supporting finite-size diagnostics rather
than claims of scalable optimization, finite-time speedup, or quantum
advantage.

\section{Model and mathematical framework}
\label{sec:model}

We consider binary optimization problems with a fixed-cardinality
constraint,
\begin{equation}
x_i\in\{0,1\},
\qquad
\sum_{i=1}^{N}x_i=K,
\end{equation}
and use as the primary test family the dense quadratic objective
\begin{equation}
C(x)=q\,x^T W x-\mu^T x,
\qquad
W\succeq0,
\qquad
q>0.
\label{eq:cost}
\end{equation}
This mean--variance-inspired form provides a controlled family of
constrained quadratic landscapes rather than a financial application
study \cite{markowitz1952,venturelli2019,grant2021}. The constraint
restricts the dynamics to
\begin{equation}
\mathcal H_K
=
\operatorname{span}
\left\{
|x\rangle:\sum_i x_i=K
\right\},
\qquad
D_K=\binom{N}{K},
\label{eq:sector}
\end{equation}
and all spectra and gaps considered below are evaluated within this
sector.
For numerical comparison we use the dimensionless problem energy
\begin{equation}
\begin{gathered}
h(x)
=
\frac{C(x)-C_{\min}}{R_C},
\\
R_C=C_{\max}-C_{\min}>0,
\qquad
P=\operatorname{diag}h,
\end{gathered}
\label{eq:pnorm}
\end{equation}
where the extrema are taken over the feasible sector and are held fixed
while the driver is varied. Thus changes of the annealing gap reflect
changes of the driver rather than changes of the problem-energy scale.
The exact feasible extrema are available in the small systems studied
here and are used only as a benchmark convention.

Constraint-preserving exchange drivers.
Let $E_{\rm allow}$ denote the set of allowed physical pairs and
$e=\{i,j\}$ one such pair. We associate with it the exchange Laplacian
\begin{equation}
L_e
=
\frac12
\left(
I-Z_iZ_j-X_iX_j-Y_iY_j
\right)
=
I-\operatorname{SWAP}_{ij},
\label{eq:Le}
\end{equation}
and construct
\begin{equation}
A(w)=\sum_{e\in E_{\rm allow}}w_e L_e,
\qquad
w_e\geq0.
\label{eq:driver}
\end{equation}
On the subspace spanned by $|01\rangle$ and $|10\rangle$, $L_e$ has
the Laplacian form
\[
\begin{pmatrix}
1&-1\\
-1&1
\end{pmatrix},
\]
while it vanishes on $|00\rangle$ and $|11\rangle$. Hence $A(w)$
preserves total excitation number and acts within $\mathcal H_K$.

A physical edge has a many-body meaning: it generates every feasible
configuration swap in which an excitation is transferred across that
pair. Consequently, within $\mathcal H_K$, $A(w)$ is the weighted
Laplacian of the induced feasible-state swap graph. If the physical
support is connected and $1\leq K\leq N-1$, this swap graph is
connected and has the unique ground state
\begin{equation}
|D_K^N\rangle
=
D_K^{-1/2}
\sum_{x\in\mathcal F_{N,K}}|x\rangle .
\label{eq:dicke}
\end{equation}
All connected candidate drivers therefore start from the same feasible
Dicke state. This common initial state is important: changing the
exchange graph modifies the kinetic structure of the annealing path
without changing its initial ground state.

Figure~\ref{fig:graphs} illustrates the distinction between the
physical exchange graph and the much larger configuration-space graph
that it induces.

\begin{figure*}[t]
\centering
\includegraphics[width=\linewidth]
{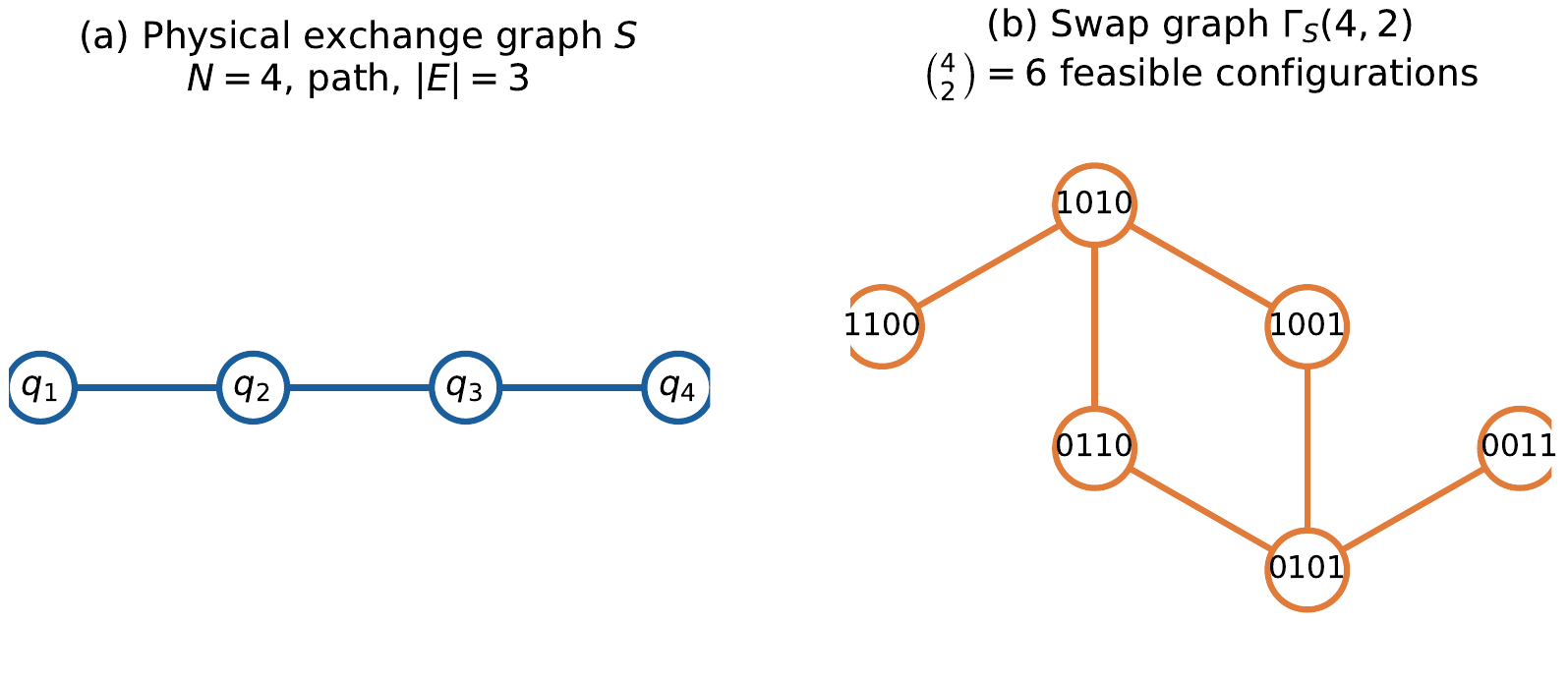}
\caption{
A physical exchange graph and its induced fixed-cardinality swap
graph. Physical vertices represent binary variables, whereas
configuration-space vertices represent feasible strings. A single
physical edge generates all allowed configuration swaps across that
variable pair.
}
\label{fig:graphs}
\end{figure*}

The discrete driver-design class contains connected unit-weight
supports with a fixed edge budget,
\begin{equation}
\begin{gathered}
\mathcal S_B
=
\left\{
S\subseteq E_{\rm allow}:
|S|=B,\;
(V,S)\ {\rm connected}
\right\},
\\
w_e=\mathbf 1_{e\in S}.
\end{gathered}
\label{eq:support_class}
\end{equation}
Different labelled supports are retained because the problem
Hamiltonian is fixed. For the numerical benchmarks, \(E_{\rm allow}\) contains all unordered
physical pairs. Thus the resource constraint fixes the number \(B\) of
supported exchanges but does not impose a particular hardware
connectivity mask.

To compare different supports on the same driver-energy scale, we
normalize by the sector bandwidth,
\begin{equation}
D(w)=\frac{A(w)}{\Lambda(w)},
\qquad
\Lambda(w)=\lambda_{\max}[A(w)],
\label{eq:normalized_driver}
\end{equation}
and define the annealing path
\begin{equation}
H(s,w)
=
(1-s)D(w)+sP,
\qquad
0\leq s\leq1.
\label{eq:annealing_path}
\end{equation}
The quantity to be optimized is the minimum spectral gap,
\begin{equation}
\begin{gathered}
\mathcal J(w)
=
\min_{s\in[0,1]}
\Delta(s,w),
\\
\Delta(s,w)=E_1(s,w)-E_0(s,w),
\end{gathered}
\label{eq:J}
\end{equation}
with discrete fixed-budget reference
\begin{equation}
\mathcal J_B^\star(P)
=
\max_{S\in\mathcal S_B}
\mathcal J(\mathbf 1_S).
\label{eq:Jstar}
\end{equation}
Because a uniform rescaling of all $w_e$ leaves $D(w)$ unchanged, the
normalization separates the overall driver scale from the question of
how its exchange channels are distributed.

\section{Local response of the minimum gap}
\label{sec:sensitivity}

The central question is now local: given a driver $w$, which exchange
modification most strongly raises its current minimum gap?

Let $s^\star$ denote an active minimum of $\Delta(s,w)$ and define,
for real normalized low-energy eigenstates,
\begin{equation}
G_e(s)
=
\langle\psi_1(s)|L_e|\psi_1(s)\rangle
-
\langle\psi_0(s)|L_e|\psi_0(s)\rangle .
\label{eq:G}
\end{equation}
Since $L_e$ is itself a graph Laplacian,
\begin{equation}
\langle\psi|L_e|\psi\rangle
=
\sum_{\substack{\{x,y\}\\x\leftrightarrow_e y}}
|\psi(x)-\psi(y)|^2,
\label{eq:dirichlet}
\end{equation}
where each unordered configuration pair is counted once. Thus $G_e$
is a Dirichlet-energy contrast between the first excited and ground
states on the set of swaps generated by physical edge $e$. It measures
how strongly that exchange channel distinguishes the two states
forming the current low-energy bottleneck.

For a regular bottleneck this interpretation can be made exact.

\begin{proposition}[Normalized minimum-gap sensitivity]
Suppose that $\Lambda(w)$ is simple, the minimum gap has a unique
minimizer $s^\star$, and the two relevant eigenvalues are simple in a
neighborhood of $(s^\star,w)$. Then
\begin{equation}
\frac{\partial\mathcal J}{\partial w_e}
=
\frac{1-s^\star}{\Lambda}
\left[
G_e(s^\star)
-
\ell_e\,\delta_D(s^\star)
\right],
\label{eq:normgrad}
\end{equation}
where
\begin{equation}
\ell_e
=
\langle\phi_\Lambda|L_e|\phi_\Lambda\rangle,
\qquad
\delta_D
=
\langle\psi_1|D|\psi_1\rangle
-
\langle\psi_0|D|\psi_0\rangle .
\end{equation}
Here \(|\phi_\Lambda\rangle\) is the normalized eigenvector of
\(A(w)\) associated with its largest eigenvalue \(\Lambda(w)\).
\end{proposition}

\begin{proof}
Hellmann--Feynman gives
$\partial_e\Lambda=\ell_e$ and therefore
\[
\partial_eD
=
\frac{L_e-\ell_eD}{\Lambda}.
\]
Applying the same theorem to $E_1-E_0$ at fixed $s$ gives the
directional gap response. The envelope theorem then evaluates this
response at the unique active minimizer $s^\star$ without requiring
the derivative of $s^\star$ itself \cite{danskin1967}.
\end{proof}

The prefactor \(1-s^\star\) also has an immediate consequence.
The instantaneous gap at \(s=1\) is independent of the driver,
so \(\Delta'_w(1,w;d)=0\). When the endpoint is the unique active
minimum, the regular minimum-gap derivative therefore vanishes;
if other minima are co-active, the directional response follows
the lower-envelope rule in Eq.~\eqref{eq:active}. In either case, \(\Delta(1)\)
provides an endpoint ceiling on \(\mathcal{J}\), so attaining this
ceiling naturally leads to abstention in the local driver-update
rule.

The same construction extends to degenerate driver bandwidths. If
$U_\Lambda$ spans the top eigenspace of $A$ and
\begin{equation}
\delta A=\sum_e d_eL_e,
\end{equation}
then
\begin{align}
\dot\Lambda[\delta A]
&=
\lambda_{\max}
\left(
U_\Lambda^\dagger\delta A U_\Lambda
\right),
\label{eq:dlambda}\\
\dot D[\delta A]
&=
\frac{\delta A}{\Lambda}
-
\frac{A}{\Lambda^2}
\dot\Lambda[\delta A].
\label{eq:dD}
\end{align}
For simple low-energy levels the fixed-$s$ gap response is
\begin{equation}
\Delta'_w(s,w;d)
=
(1-s)
\left[
\langle\psi_1|\dot D|\psi_1\rangle
-
\langle\psi_0|\dot D|\psi_0\rangle
\right].
\label{eq:dirgap}
\end{equation}
If the ground level is simple but the first-excited level is
degenerate, let \(U_1\) span the corresponding first-excited
eigenspace and define \(V=(1-s)\dot D\). The ordered
first-excited-level response is then obtained by diagonalizing
the perturbation within this eigenspace,
\[
E_1'=\lambda_{\min}\!\left(U_1^\dagger V U_1\right),
\qquad
E_0'=\langle\psi_0|V|\psi_0\rangle,
\]
so that \(\Delta'_w=E_1'-E_0'\). An unresolved ground-state
degeneracy lies outside the regular sensitivity assumptions used here.
This form is particularly important for a discrete edge replacement.
If supported edge $r$ is removed and absent edge $a$ is added, the
complete proposal direction is
\begin{equation}
\delta A=L_a-L_r.
\label{eq:replacement_direction}
\end{equation}
At a degenerate bandwidth the directional derivative of $\Lambda$ is
not a linear functional of $d$. The replacement must therefore be
scored as one complete direction; independently scoring the addition
and removal and subtracting the two numbers is not generally
equivalent.

The derivative is nevertheless a \emph{local} statement. A complete
replacement is a finite move. To distinguish these two notions, we
parameterize the direction selected at a starting support $S_0$ as
\begin{equation}
A(\alpha)
=
A(S_0)
+
\alpha(L_a-L_r),
\qquad
0\leq\alpha\leq1.
\label{eq:alpha_path_framework}
\end{equation}
For $\alpha\ll1$ this probes the differential prediction, whereas
$\alpha=1$ produces the complete unit-edge replacement. Intermediate
values are weighted diagnostic drivers and need not belong to the
discrete class $\mathcal S_B$. This distinction is central to the
numerical analysis below.

\section{Bottleneck migration and relinearization}
\label{sec:bottleneck_framework}

The minimum-gap derivative is evaluated at the bottleneck that is
active \emph{at the current driver}. As the driver changes, however,
the identity and location of that bottleneck can also change. This
provides the basic mechanism by which a locally favorable direction
can cease to be favorable at finite distance.

Suppose several smooth local minima of $\Delta(s,w)$ coexist. If
$\mathcal A(w)$ denotes the set of active minima, then under the
appropriate smooth-envelope conditions,
\begin{equation}
\mathcal J'(w;d)
=
\min_{s\in\mathcal A(w)}
\Delta'_w(s,w;d).
\label{eq:active}
\end{equation}
Thus a direction that improves one minimum need not improve another.
In particular, raising the currently limiting minimum can expose a
previously higher minimum, which then becomes the new global
bottleneck. We refer to this exchange of the active minimum as
\emph{bottleneck migration}.

This observation also suggests the appropriate local response when two
minima become co-active. Let
\[
\mathcal{A}_c
=
\{b:\Delta(s_b,w_c)=\mathcal{J}(w_c)\}
\]
denote the set of active branch indices at the resolved exchange, where
\(s_b\) is the annealing coordinate of branch \(b\). Define
\[
g_b(d)=\Delta'_w(s_b,w_c;d).
\]
Rather than continuing the direction chosen for the old bottleneck, we
can relinearize at the current driver and choose
\begin{equation}
d_{\rm act}
\in
\operatorname*{arg\,max}_{d\in\mathcal D_c\cup\{0\}}
\min_{b\in\mathcal A_c}g_b(d).
\label{eq:adaptive_active_menu}
\end{equation}
Here
\[
\mathcal{D}_c
=
\{\mathbf{e}_a-\mathbf{e}_r:
a,r\in\operatorname{supp}(w_c),\;a\neq r\}
\]
is the set of ordered unit weight-transfer directions between currently
supported exchanges. Here $\mathbf e_e$ denotes the standard basis vector associated with physical edge $e$ in coupling-weight space. The zero direction allows abstention. This max--min rule protects all
currently active minima to first order: among the available directions,
it chooses the one with the largest worst response across the tied
bottlenecks.

Eq.~\eqref{eq:adaptive_active_menu} is not a claim of global
optimization. It is the natural local continuation of
Eq.~\eqref{eq:active}: once the active bottleneck changes, the
linearization should also be updated. The numerical relinearization
experiment below tests precisely this proposition on the resolved
bottleneck exchanges.

Low-energy reorganization accompanying such migration is quantified
through the rank-two projector
\begin{equation}
\Pi
=
|\psi_0\rangle\langle\psi_0|
+
|\psi_1\rangle\langle\psi_1|,
\end{equation}
and the normalized subspace overlap
\begin{equation}
O_{\rm sub}
=
\frac12
\operatorname{Tr}
\left(
\Pi_{\rm before}\Pi_{\rm after}
\right).
\label{eq:subspace_overlap_framework}
\end{equation}
An overlap close to one indicates that essentially the same
two-dimensional low-energy sector remains relevant, whereas a smaller
value signals stronger reorganization. This provides a state-space
diagnostic complementary to tracking the position of the gap minimum
itself. The rank-two overlap is interpreted only when the two lowest
levels form an isolated two-dimensional spectral subspace. If the
first-excited level is degenerate with a higher level, the rank-two
projector is not uniquely defined and a higher-rank subspace
treatment is required.

\section{Spectral and problem-informed approximations}
\label{sec:limits}

The exact sensitivity requires the current quantum low-energy states.
We therefore consider two limiting sources of classical information
that can be used to approximate favorable exchange directions.

At the beginning of the anneal, the symmetric exclusion-process gap
identity following from the interchange-process result of Caputo,
Liggett, and Richthammer \cite{caputo2010} gives
\begin{equation}
\operatorname{gap}A(w)
=
\lambda_2[L_S(w)],
\qquad
1\leq K\leq N-1,
\label{eq:exclusion}
\end{equation}
where $L_S$ is the physical $N$-vertex graph Laplacian. Consequently,
\begin{equation}
\Delta(0,w)
=
C(S,w)
=
\frac{\lambda_2[L_S(w)]}{\Lambda(w)}.
\label{eq:certificate}
\end{equation}
This is an exact initial-gap descriptor, but only an upper bound on the
minimum gap,
\[
\mathcal J(w)\leq C(S,w).
\]
It therefore characterizes initial spectral connectivity rather than
the interior bottleneck itself.

Its normalized directional response is
\begin{equation}
C'(w;d)
=
\frac{\lambda_2'(w;d)}{\Lambda}
-
\frac{\lambda_2}{\Lambda^2}
\dot\Lambda[\delta A].
\label{eq:Cdir}
\end{equation}
For a simple Fiedler eigenvalue with normalized eigenvector $v$, an
individual unnormalized edge derivative is $(v_i-v_j)^2$. For a
degenerate Fiedler eigenspace $U_2$,
\begin{equation}
\lambda_2'(w;d)
=
\lambda_{\min}
\left(
U_2^\dagger\delta L_S U_2
\right).
\end{equation}
A positive rank-one edge addition cannot raise a degenerate
unnormalized algebraic connectivity to first order. This creates a
specific limitation of single-edge spectral ranking on the cycle and
star-plus-edge geometries considered below; the path-plus-chord
geometry, with a simple Fiedler eigenvalue, provides a nondegenerate
control.

At an interior avoided crossing, a different approximation becomes
natural. Suppose the two relevant states are dominated by two
configurations,
\begin{equation}
|\psi_0\rangle
\simeq
\cos\theta|a\rangle+\sin\theta|b\rangle,
\qquad
|\psi_1\rangle
\simeq
-\sin\theta|a\rangle+\cos\theta|b\rangle.
\label{eq:two}
\end{equation}
Define the diagonal exchange matrix elements
\[
d_a^{(e)}=\langle a|L_e|a\rangle,
\qquad
d_b^{(e)}=\langle b|L_e|b\rangle.
\]
An exchange directly connecting $a$ and $b$ then gives
\begin{equation}
G_e\simeq2\sin2\theta,
\end{equation}
whereas an exchange diagonal in this restricted basis gives
\begin{equation}
G_e
\simeq
\cos2\theta
\left(
d_b^{(e)}-d_a^{(e)}
\right).
\end{equation}

At equal mixing the direct exchange reaches $G_e=2$, while the
diagonal-only contribution vanishes within the two-state ansatz. This
suggests that exchanges coupling the competing low-energy regions can
be especially effective near a localized avoided crossing.

The argument is conditional: the relevant competitor need not be the
second-lowest classical configuration, equal mixing need not occur at
the minimum gap, and additional states may contribute. The spectral
and crossing pictures should therefore be viewed as two limits of the
same exchange sensitivity rather than as universally ordered stages of
an anneal.

\section{Classical approximations to the low-energy direction}
\label{sec:surrogates}

For the dense quadratic family we use two problem-informed
approximations to the crossing-relevant direction.

The first is a relax--swap construction. We solve the convex
relaxation
\begin{equation}
\min_z
qz^TWz-\mu^Tz,
\qquad
0\leq z_i\leq1,
\qquad
\sum_i z_i=K,
\label{eq:relax}
\end{equation}
and round the solution to a feasible $\widehat x$. For an occupied
$i$ and unoccupied $j$, let $\delta=e_j-e_i$. The exact change under
this classical swap is
\begin{equation}
\delta C_{ij}
=
q
\left(
2\delta^T W\widehat x+\delta^T W\delta
\right)
-
\mu^T\delta.
\label{eq:swapcost}
\end{equation}
The pair minimizing this change defines the singleton or relax--swap
proposal. This construction uses no exact combinatorial optimum, but
recovering a favorable classical swap is not equivalent to maximizing
the quantum gap sensitivity.

The second approximation replaces individual configurations by
classical basins. Deterministic steepest descent on the current swap
graph assigns each feasible configuration to a terminal state. Starting
from the rounded relaxation and 20 uniformly sampled feasible
configurations, the two lowest-cost distinct discovered terminals
define basins $A$ and $B$.

Let $u_A$ and $u_B$ be nonnegative normalized amplitudes supported on
these disjoint basins. Replacing the two configurations in
Eq.~\eqref{eq:two} by these basin states gives
\begin{equation}
G_e
\simeq
\cos2\theta\,(b_e-a_e)
-
2\sin2\theta\,t_e,
\label{eq:generalflux}
\end{equation}
where
\[
\begin{gathered}
a_e=\langle u_A|L_e|u_A\rangle,\quad b_e=\langle u_B|L_e|u_B\rangle,\\
t_e=\langle u_A|L_e|u_B\rangle .
\end{gathered}
\]
At equal mixing the relevant inter-basin contribution becomes
\begin{equation}
F_e(A,B)
=
-2t_e
=
2
\sum_{\substack{x\in A,\;y\in B\\x\leftrightarrow_e y}}
u_A(x)u_B(y).
\label{eq:flux}
\end{equation}
We refer to $F_e$ as the amplitude weighted inter-basin kinetic coupling. The historical label ``flux'' appearing in archived data and plot legends refers to this quantity, not to a quantum probability current.

The amplitudes used operationally are
\begin{equation}
\begin{gathered}
u_A(x)
=
\frac{
\mathbf1_A(x)
e^{-\gamma[h(x)-h_A]/2}
}{
\left(
\sum_{z\in A}
e^{-\gamma[h(z)-h_A]}
\right)^{1/2}
},
\\
h_A=\min_{x\in A}h(x),
\qquad
\gamma=8,
\end{gathered}
\label{eq:amplitude}
\end{equation}
and analogously for $B$. The choice is a model for the internal basin amplitudes, not a derived eigenstate.
This approximation leaves several potential sources of error:
identifying the pair of relevant basins, modelling the internal amplitudes, assuming the mixing to be close to equal, neglecting additional states, and substituting the total normalised sensitivity by its direct kinetic component. To isolate the first of these, we define a
retrospective oracle-pair reference from
\begin{equation}
p(x)
=
\frac{
|\psi_0(x)|^2+|\psi_1(x)|^2
}{2}
\end{equation}
at the computed minimum and select the two basins with the largest
summed $p$. This construction uses quantum information and is employed
only diagnostically.

\section{Numerical protocol}
\label{sec:protocol}

The numerical experiments are arranged to probe progressively higher levels of the framework. First the exact local sensitivity, then its finite-step breakdown via bottleneck migration, then relinearization at a resolved exchange and finally the performance of finite discrete search and classical approximations.

The main dense ensemble uses \begin{equation} W= \frac{ MM^T+\operatorname{diag}u }{ \max_i(MM^T+\operatorname{diag}u)_{ii} }, \end{equation} where $M$ has independent standard normal entries, $u_i\sim U(0.05,0.25)$, $\mu_i\sim U(0.1,1.0)$, and $q\sim U(0.4,1.0)$. Only objectives with a unique feasible optimum that is separated from the next cost by more than $10^{-8}$ are retained.
The matrix dimensions, random generators, labelled supports, starting-support schedule, gap locators and numerical tolerances used in these calculations are given in Appendix~\ref{app:repro}.

The exhaustive study uses 30 matched $N=6$, $K=3$ objectives for
budgets $B=5,\ldots,9$. The numbers of connected labelled supports are
1296, 3660, 5700, 6165, and 4945, respectively, giving
\GraphCount{} graphs and 652,980 graph--instance evaluations. Six
starting supports per objective and budget give
\[
30\times5\times6=900
\]
starting graph--instance records. The oracle-sensitivity and migration
experiments reuse these same starts.

The numerical reference $\widehat{\mathcal J}$ is obtained by scanning
the annealing coordinate and refining the smallest observed gap.
Graph coverage in the exhaustive study is combinatorially complete,
whereas minimization over the continuous annealing coordinate remains
numerical.

For the oracle test, the current minimum is refined before any finite
candidate gap is evaluated. Every legal connected one-edge replacement
is scored from the normalized directional response. For the fixed oracle
paths in Figures~\ref{fig:oracle_steps}--\ref{fig:subspaces} and the
corresponding migration and relinearization calculations, scores within an absolute
tolerance of $10^{-12}$ of the largest score are treated as numerically
tied. Among these candidates, the replacement with the smallest
support mask is selected. Physical edges are indexed from zero in
lexicographic order, and the mask of a proposed support $S$ is
$m(S)=\sum_{e\in S}2^{\operatorname{index}(e)}$.
This deterministic tie rule uses no finite candidate-gap outcomes.
Cases whose retained
minimum occurs at $s^\star=1$, or for which no positive direction is
available, lead to abstention. The selected direction is then followed
along Eq.~\eqref{eq:alpha_path_framework}; hence finite outcomes do not
enter the selection of the direction.
For the fractional-step continuation tests, the re-minimized gaps are
evaluated with the refined gap locator and an absolute improvement
tolerance of \(10^{-10}\). At the complete replacement
\(\alpha=1\), both the starting and final supports are evaluated with
the production scan-and-zoom convention, using an absolute improvement
tolerance of \(10^{-12}\).

Along these paths we track all numerically resolved local minima of
$\Delta(s)$. A resolved bottleneck exchange requires two coexisting
local minima, separated by at least $0.02$ in $s$, to reverse their
depth ordering. Continuous motion of a single minimum, including
motion into an endpoint, is treated separately.

For each overlap comparison we evaluate the three lowest levels at
the active minima of both endpoint drivers and require
$E_2-E_1>10^{-9}$ at each point. Invalid rank-two comparisons are
flagged and excluded from overlap summaries. All 798 selected dense
paths pass both checks, including the 19 with an initial minimum at
$s=0$; none is excluded. Appendix~\ref{app:isolation} gives the measured
minimum separations and the failure-handling rule.

At resolved interior exchanges, the relinearization experiment compares
three directions: continuation of the original update, optimization
with respect to the incoming minimum alone, and the active-minima
max--min rule of Eq.~\eqref{eq:adaptive_active_menu}. Each is evaluated
with the same local step
\begin{equation}
t=\min(\epsilon,w_{\min}/4),
\qquad
\epsilon=10^{-2}.
\label{eq:relinearization_step}
\end{equation}
Here
\[
w_{\min}
=
\min_{e\in\operatorname{supp}(w_c)} (w_c)_e
\]
is the smallest positive coupling weight at the switching driver. This is a continuous local reweighting experiment on the expanded
support and not a completed unit-weight graph replacement.

The finite production search operates instead in the discrete class
$\mathcal S_B$. A legal move replaces one supported edge by one absent
edge while preserving connectivity. Guided search ranks complete
replacement directions, evaluates the finite gaps of the top $k=8$
candidates, and accepts the best improvement above $10^{-12}$. The
all-neighbor reference evaluates every legal replacement. The matched random control begins from the same supports and is
restricted to the finite candidate-gap evaluation cap of each matched
guided restart. Each random realization first takes the best final
gap across the six starts and divides it by the exhaustive reference.
The resulting 200 best-of-six ratios are then averaged within each
objective and paired with that objective's guided best-of-six ratio.
The random hit percentage averages the individual realization's hit
indicator, defined by $R\geq0.999999$, across all 200 realizations and
30 objectives. It does not threshold the averaged ratio.
This is an evaluation-cap match; guided search additionally computes
the sensitivity used to rank candidates. Appendix~\ref{app:random}
specifies the batching, stopping rules, seed mapping, exact aggregation
formulas, and 20,000 objective-bootstrap resamples.

For objective $i$, discrete-search performance is measured by
\begin{equation}
R_i
=
\frac{
\max_r
\widehat{\mathcal J}^{\rm final}_{i,r}
}{
\max_{S\in\mathcal S_B}
\widehat{\mathcal J}_i(S)
},
\label{eq:R}
\end{equation}
and a numerical optimum hit is recorded when
$R_i\geq0.999999$.

The local-selector study uses 30 accepted objectives in each of nine
$(N,\mathrm{geometry})$ cells, with
$N=7,8,10$, $K=\lfloor N/2\rfloor$, and $B=N$. The three physical
geometries are the cycle, star-plus-edge, and path-plus-chord.
Acceptance requires a one-exchange separation of the two lowest
classical configurations, successful discovery of two basins,
$s^\star\neq1$, and positive maximal raw and normalized sensitivities.
Practical and spectral selector choices are fixed before the refined
quantum diagnostics are recomputed.

For this comparison, every single-edge weight-increase direction $+\mathbf e_e$ is admissible and its quality is
\begin{equation}
Q_i(e)
=
\frac{g_e}{\max_j g_j},
\qquad
g_e=\mathcal J'(w;\mathbf e_e),
\qquad
\max_j g_j>0.
\label{eq:Q}
\end{equation}
Thus $Q$ measures how well a selector identifies the best available
\emph{local} normalized-sensitivity direction; it is not a finite-step
performance ratio.
Finally, to test whether the sensitivity--migration mechanism depends
on the dense quadratic generator, we use 30 independent sparse
mixed-sign frustrated Ising objectives on randomly relabeled
triangular-prism interaction graphs. These use the same $N=6$, $K=3$
sector, starting-support schedule, driver normalization, and oracle
protocol. Only the local sensitivity and migration analyses are
repeated for this second family.

The guided--random comparison in Table~\ref{tab:random} is paired
within each of the 30 objectives. Its 95\% percentile bootstrap
intervals resample objectives, after averaging the 200 random
realizations within each objective, as specified in
Appendix~\ref{app:random}. These intervals are pointwise across budgets.
Repeated starts and migration events are not independent objectives;
the event counts and the selector medians are descriptive summaries.
The selector summaries are conditional on the acceptance criteria
above. An interval containing zero is not evidence of equivalence.

\section{Results}
\label{sec:results}

The numerical results are organized around a single question. The
exchange-resolved sensitivity derived above is a local statement: it
identifies how the active minimum gap responds to an infinitesimal
change of the driver. A finite driver modification, however, changes
the Hamiltonian path itself and can alter which avoided crossing
controls the minimum gap. We therefore proceed in four stages. We first
verify the local sensitivity directly. We then follow the selected
direction to finite strength and identify how its predictive power is
lost. Next we test whether the resulting bottleneck-migration mechanism
persists in a different objective family. Finally, we ask whether
relinearizing the response when competing minima become active can
repair the finite-step failure. The exhaustive graph search and
classical-selector studies are then used as supporting tests of how
much of this local information remains useful when converted into
discrete or approximate guidance.

\subsection{From local sensitivity to finite driver changes}
\label{sec:validation}
\label{sec:oracle}

We begin with the point at which the mathematical construction makes
its sharpest prediction. For a given starting driver, the normalized
directional sensitivity is evaluated at the current minimum-gap
bottleneck, and the legal connected one-edge replacement with the
largest positive directional response is selected. No finite candidate
gap is used in making this choice. The subsequent finite evolution of
the selected direction therefore tests the distinction between the
local derivative and the finite driver modification rather than an
accept--reject search procedure.

The test uses the 900 starting graph--instance records of the dense
quadratic cohort,
\[
30\ \text{objectives}\times5\ \text{budgets}\times6\ \text{starts}
=900.
\]
For 102 starts the retained minimum occurs at $s^\star=1$, where the
driver cannot raise the endpoint gap, and the oracle therefore
abstains. The remaining 798 starts, consisting of 779 interior minima
and 19 minima at $s=0$, define oracle-selected update paths. If $r$ is
the removed edge and $a$ the added edge, the selected direction is
continued according to
\begin{equation}
\begin{gathered}
A(\alpha)
=
A(S_0)+\alpha(L_a-L_r),
\\
D(\alpha)
=
\frac{A(\alpha)}
{\lambda_{\max}[A(\alpha)]},
\qquad
0<\alpha\leq1 .
\end{gathered}
\label{eq:oracle_alpha_path}
\end{equation}
The limit $\alpha\rightarrow0^+$ probes the local derivative, whereas
$\alpha=1$ is the complete unit-edge replacement. Intermediate values
are weighted $B+1$-support drivers and are used only to diagnose the
continuation from the local to the finite regime.

The infinitesimal prediction is borne out directly. At both
$\alpha=10^{-3}$ and $10^{-2}$, all 798 selected directions increase
the re-minimized gap. This confirms that the selected directions improve the numerically
re-minimized gap at the two smallest tested step sizes.

The more informative behavior appears as the same direction is
continued away from the linear regime. The fractions of improving
paths are 99.75\%, 99.37\%, 99.00\%, and 96.87\% at
$\alpha=0.05$, $0.10$, $0.25$, and $0.50$, respectively. At the
complete replacement, 662 of the 798 paths improve ($82.96\%$). The evolution
is shown in Fig.~\ref{fig:oracle_steps}.

\begin{figure}[htbp]
\centering
\includegraphics[width=\linewidth]
{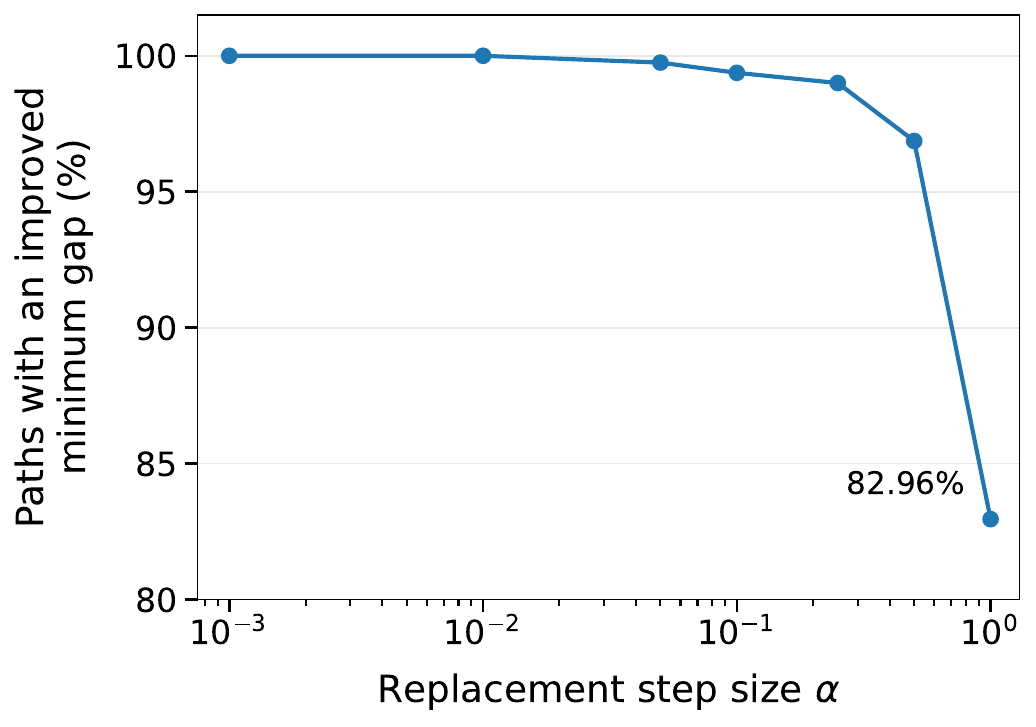}
\caption{
Oracle-selected directions from local perturbations to complete edge
replacements. The 798 non-abstained starts each define one fixed
replacement path $S_0\rightarrow S_1$ parameterized by $\alpha$.
All selected directions improve the re-minimized gap at
$\alpha=10^{-3}$ and $10^{-2}$, whereas the full replacement at
$\alpha=1$ improves $662/798=82.96\%$ of paths. Fractional
$\alpha$ values are weighted diagnostic drivers outside the
unit-weight $B$-edge class. Finite candidate gaps are not used to
select the direction; the line connects the tested step sizes as a
visual guide.
}
\label{fig:oracle_steps}
\end{figure}

This distinction is central. The local sensitivity is highly reliable
where it is mathematically local, but it is not a finite-replacement
guarantee. The full-replacement frequency in Figure~\ref{fig:oracle_steps}
uses all 798 non-abstained paths as its denominator. Among them, 797
have at least one improving complete neighbor at the $10^{-12}$
production-gap tolerance. Thus the same 662 successful replacements
give $662/798=82.96\%$ overall and $662/797=83.06\%$ conditional on an
improving complete neighbor existing. The remaining selected start
has a positive local direction and beneficial small fractional steps,
but no improving complete edge replacement. This explicitly
illustrates the difference between a locally beneficial reweighting
and a beneficial complete replacement; its parameters and gaps are
given in Appendix~\ref{app:isolation}.

For a start with a positive best-neighbor improvement, the fraction
of that improvement captured by the selected replacement is defined as
\begin{equation}
\begin{gathered}
\eta_{\rm move}
=
\frac{
\widehat{\J}(S_{\rm sens})-\widehat{\J}(S_0)
}{
\widehat{\J}_{\rm best}-\widehat{\J}(S_0)
},
\\
\widehat{\J}_{\rm best}
=
\max_{T\in\{S_0\}\cup\mathcal N(S_0)}
\widehat{\J}(T),
\end{gathered}
\label{eq:eta_move}
\end{equation}
Here $\mathcal N(S_0)$ denotes the set of legal connected supports
obtained from $S_0$ by one complete edge replacement. This ratio is
evaluated only when
$\widehat{\J}_{\rm best}-\widehat{\J}(S_0)>10^{-12}$.
Its pooled median over the 797 eligible starts is 0.892.

Endpoint-limited cases numbering 102 are useful for verifying consistency.
The oracle declines on all of them and there is no one-edge
neighbor with a legally improving case in the reorganised production space.
The key question is not whether the local derivative works; the small-step test
demonstrates that it does work in the tested ensemble. The
question is why a direction that initially raises the minimum gap can
cease to do so as the driver is changed by a finite amount.

\subsection{Bottleneck migration: repairing one minimum exposes another}
\label{sec:migration}

To answer this question, we follow the same 798 oracle-selected
directions while resolving the local minima of
$\Delta(s,\alpha)$ throughout the update. No re-ranking is performed
along the path. A bottleneck exchange is recorded when two distinct
local minima, separated by at least $0.02$ in $s$, coexist and reverse
their ordering in gap magnitude. Continuous motion of a single
minimum, including motion to an endpoint, is treated separately.

Resolved switching occurs on
\[
421/798=52.8\%
\]
of the paths, with 491 resolved exchange events in total.
Among these, 251 paths contain at least one
interior-to-interior exchange, giving 279 such events. The key
observation is the order in which the changes occur. On 420 of the 421
switching paths, the gap at the first resolved exchange is already
larger than the minimum gap of the original driver. Thus the dominant
sequence is not that the selected perturbation immediately fails.
Rather, it first does what the local sensitivity predicts: it repairs
the active bottleneck. As that bottleneck is raised, a different local
minimum of the same annealing spectrum becomes competitive and
eventually becomes the global minimum.

This process is illustrated by Figure~\ref{fig:migration} for one path. Two
local minima coexist over a finite interval of $\alpha$. The branch
that limits the original driver rises, while the competing branch
also rises but more slowly. Their crossing makes the competing branch
the new limiting gap. This example illustrates bottleneck migration,
not the separate subset of 136 events with a decreasing incoming branch.

\begin{figure}[htbp]
\centering
\includegraphics[width=\linewidth]
{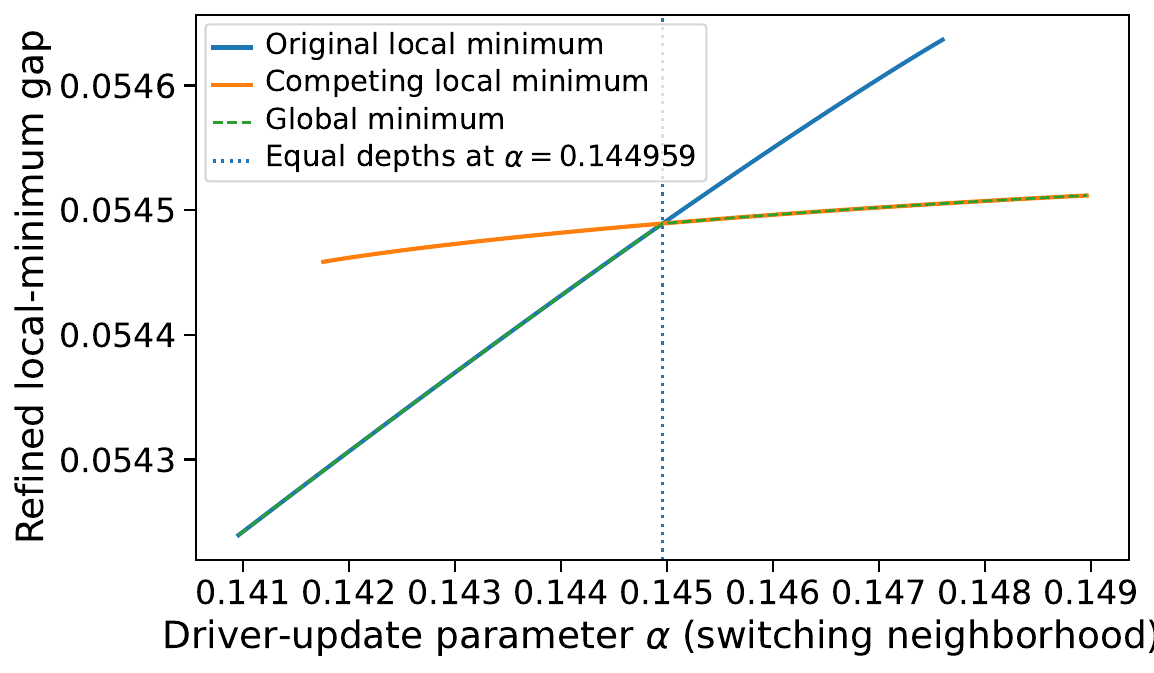}
\caption{
Example of a resolved bottleneck exchange along one oracle-selected
driver update. Two distinct local minima coexist and reverse their
depth ordering as $\alpha$ increases. In the illustrated $B=6$ path,
their gaps become equal near $\alpha=0.144959$, at annealing locations
$s\simeq0.337921$ and $0.436638$. The original bottleneck is first
raised, after which the competing minimum becomes globally limiting.
Both branches increase in the displayed neighborhood; the incoming
branch becomes limiting because it increases more slowly.
This ``repair-and-expose'' behavior is the finite-update mechanism
examined in the text; cohort-level statistics use all 798
non-abstained paths.
}
\label{fig:migration}
\end{figure}

The exchange of the active minimum is accompanied by a marked change
of the low-energy sector. To quantify this without relying on the
identity of individual eigenvectors, we use the rank-two projector
\begin{equation}
\Pi
=
|\psi_0\rangle\langle\psi_0|
+
|\psi_1\rangle\langle\psi_1|,
\qquad
O_{\rm sub}
=
\frac{1}{2}
\operatorname{Tr}
\left(
\Pi_{\rm before}\Pi_{\rm after}
\right).
\label{eq:subspace_overlap}
\end{equation}
Identical two-dimensional low-energy subspaces give $O_{\rm sub}=1$.
The endpoint spectra satisfy the isolation test in
Appendix~\ref{app:isolation} for every comparison.
The median initial-to-final overlap is 0.449 on switching paths,
compared with 0.915 on paths without a resolved switch
(Fig.~\ref{fig:subspaces}). The two path groups are not independent
randomized populations, so this comparison is descriptive rather than
causal. It nevertheless shows that resolved bottleneck migration is
associated with substantially stronger reorganization of the
low-energy subspace.
\begin{figure}[tb]
\centering
\includegraphics[width=\linewidth]
{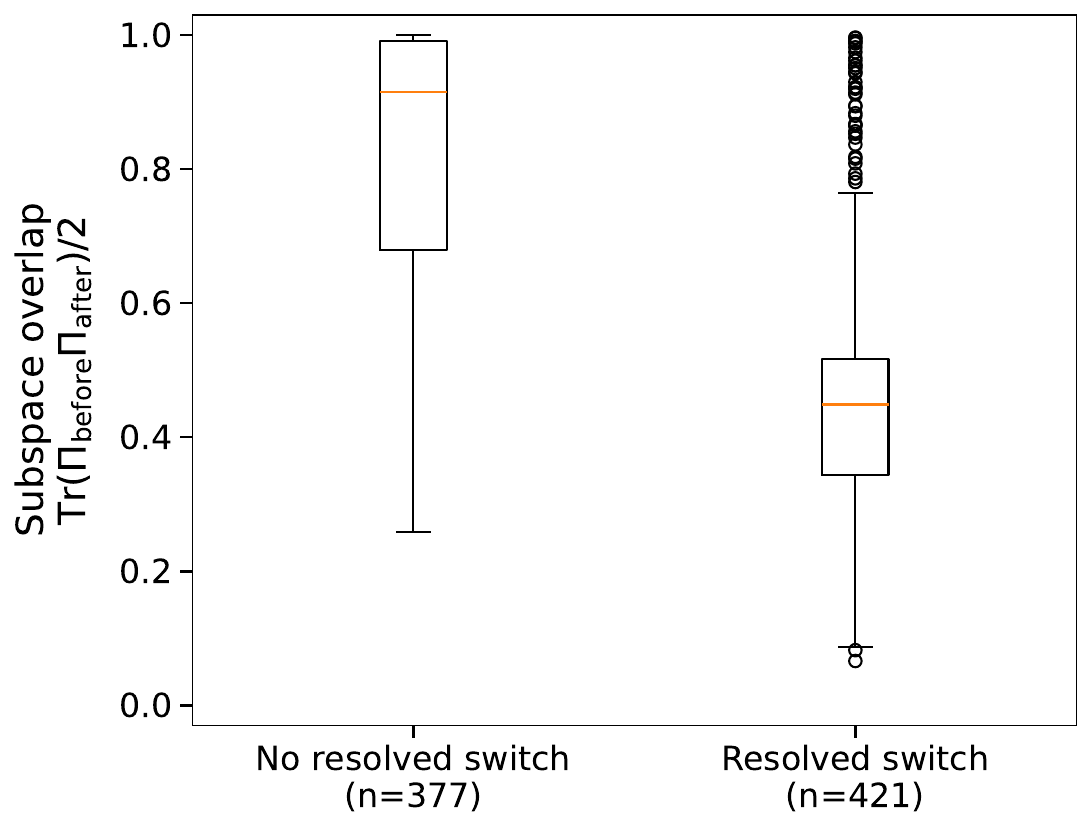}
\caption{
Initial-to-final low-energy subspace overlap, evaluated at each
driver's active minimum, for 377 paths without a resolved switch and
421 switching paths. The medians are 0.915 and 0.449, respectively.
The path groups share the same 30 objectives, so their separation is
an association rather than an independent-sample causal comparison.
Both the driver and the active annealing coordinate change in this
diagnostic. All 798 comparisons satisfy $E_2-E_1>10^{-9}$ at both
evaluation points; no rank-two comparison was excluded.
}
\label{fig:subspaces}
\end{figure}
The branch-resolved derivatives make the mechanism still more explicit. At 136 of the 279 interior-to-interior exchange events, the departing branch has a positive derivative along the original update direction while the incoming branch has a negative derivative. Thus the original direction is still locally doing the job for which it was selected: it continues to raise the old bottleneck. Its failure is that this bottleneck is no longer the one that determines
$\J$. Once the second branch becomes active, continuing to optimize the first branch is no longer equivalent to optimizing the lower envelope. This mechanism explains part, but not all, of the loss of finite-step reliability. Complete replacements are harmful on 101/421 switching paths (24.0\%) and 35/377 paths without a resolved switch (9.3\%). Most switching paths therefore still end with an improved gap, and some harmful paths have no resolved switch. Bottleneck migration is accordingly neither necessary nor sufficient for a harmful finite replacement. What it provides is a directly identifiable mechanism by which the locally correct direction can cease to be the appropriate direction for the global minimum gap.
\subsection{The mechanism is not specific to the dense objective family}
\label{sec:frustrated_second_family}
The repair-and-expose interpretation would be of limited significance if it resulted only from the dense quadratic objective generator. We therefore repeat the oracle sensitivity and migration tests for 30 independently generated mixed-sign frustrated Ising objectives,
\begin{equation}
\begin{gathered}
C_F(x)
=
\sum_{\{i,j\}\in E_P} J_{ij} z_i z_j
+
\sum_i f_i z_i,
\\
z_i=2x_i-1,
\qquad
\sum_i x_i=3 ,
\end{gathered}
\label{eq:frustrated_objective}
\end{equation}
where \(E_P\) is a randomly relabeled triangular-prism interaction
graph, \(J_{ij}=\pm1\) with equal probability, and
\(f_i\sim\mathcal{U}(-0.15,0.15)\). Draws are retained only when both
bond signs are present, the preferred bond relations are frustrated
so that they cannot all be satisfied simultaneously, and the feasible
objective has a unique minimum separated from the next feasible value
by more than \(10^{-8}\). We retain \(N=6\), \(K=3\), the same five driver budgets,
the same six-start schedule, the same driver normalization, and the
same oracle direction-selection principle. The numerical minimum-gap
locators are family specific: the dense-family selections use the
refined oracle locator, whereas the frustrated-family selections use
a 321-point uniform scan augmented by endpoint probes and local-minimum
refinement. This is a mechanism test only; the exhaustive-search and
classical-selector studies are not repeated for the second family.

Of the 900 starting records, 809 admit an oracle-selected direction.
Of the remaining 91 cases, 90 are endpoint-limited abstentions and one
is a numerical-cluster abstention associated with a very small retained
gap.

The finite landscape again develops resolved bottleneck exchanges, but
with a different frequency. Switching occurs on
\[
188/809=23.2\%
\]
of the frustrated-family paths, compared with 52.8\% in the dense
family. All 188 first switches occur after the minimum gap has risen
above its initial value. Moreover, switching paths again have a much
smaller median low-energy-subspace overlap than nonswitching paths,
0.478 compared with 0.962. The comparison is summarized in
Table~\ref{tab:family_comparison}.

\begin{table*}[tp]
\centering
\small
\caption{
Comparison of the local-sensitivity and bottleneck-migration
mechanisms at $N=6$, $K=3$. Each family contains 30 independently
generated objectives and 900 starts on the same driver-support
schedule. Path counts are repeated observations within objectives and
should not be interpreted as independent-objective counts.
}
\label{tab:family_comparison}
\begin{tabular}{@{}lcc@{}}
\toprule
Diagnostic & Dense quadratic & Sparse frustrated \\
\midrule
Non-abstained directions
& 798 & 809 \\
Gap increases at $\alpha=10^{-3},10^{-2}$
& 798/798 & 809/809 \\
Paths with resolved switch
& 421/798 (52.8\%) & 188/809 (23.2\%) \\
Interior--interior switching paths
& 251/798 & 81/809 \\
Median overlap: switching
& 0.449 & 0.478 \\
Median overlap: no resolved switch
& 0.915 & 0.962 \\
First switch above initial gap
& 420/421 & 188/188 \\
\bottomrule
\end{tabular}
\end{table*}

The relevant cross-family result is therefore not a universal switching
frequency. The frequency changes substantially with the objective
family. What persists is the sequence of phenomena: the local
sensitivity correctly predicts small driver changes, finite
continuation can change the active bottleneck, the first resolved
exchange generally occurs after the original gap has been raised, and
switching is accompanied by stronger low-energy reorganization. This
is the level at which the numerical evidence supports a mechanism that
is not tied to one objective generator.

\subsection{Relinearization at a bottleneck exchange}
\label{sec:oracle_relinearization}

The migration picture suggests a direct response. Once two minima
become co-active, the derivative should no longer be optimized for only
the bottleneck from which the update began. The relevant local object
is instead the lower envelope of the active gap branches. We therefore
test whether recomputing the direction at a resolved exchange and
protecting both active minima restores local improvement.

The test uses all 279 dense-family interior-to-interior exchange events,
arising from 251 paths and 27 distinct objectives. At the equality
driver $w_c$, the two resolved minima are active within a gap tolerance
of $10^{-10}$. Three directions are compared. The first simply
continues the direction chosen at the original bottleneck. The second
re-optimizes using only the newly incoming minimum. 
The third direction uses the active-minima max--min prescription defined in Eq.~\eqref{eq:adaptive_active_menu}, evaluated at the equality driver \(w_c\), where the competing minima are co-active.  Here $\mathcal A_c$ contains the active minima. The candidate set
$\mathcal D_c$ contains all ordered weight transfers between currently
supported edges, and the zero direction permits abstention if no
candidate has a sufficiently positive worst-branch response.

All three choices are evaluated with the same feasibility-capped step,
\[
t=\min(\epsilon,w_{\min}/4),
\qquad
\epsilon=10^{-2}.
\]
This is deliberately a local continuous reweighting experiment on the
expanded support. It is not a completed unit-weight graph-search move.

The distinction between the three rules is sharp
(Table~\ref{tab:adaptive_primary}). Continuing the original direction
is harmful at 171 of the 279 events. Re-optimizing only the incoming
minimum reduces the number of harmful outcomes to 115 but does not
remove them. The active-minima rule improves the gap at
278 events and abstains at the remaining event, with no harmful outcome observed.
Of the 171 events at which simple continuation is harmful, 170 become
improvements under the active-minima rule and one becomes unchanged.

\begin{table*}[tp]
\centering
\small
\caption{
One-step oracle relinearization at 279 interior bottleneck exchanges
arising from 27 contributing objectives. All rules use the same
feasibility-capped local step. Mean
$r=(J_{\rm after}-J_c)/J_c$. Counts are event-level observations and
are not 279 independent-objective trials.
}
\label{tab:adaptive_primary}
\begin{tabular}{@{}lrrrr@{}}
\toprule
Direction rule & Improve & Harm & Unchanged & Mean $r$\\
\midrule
Continue original
& 108 & 171 & 0 & $-0.002143$\\
Incoming minimum only
& 164 & 115 & 0 & $+0.000239$\\
Both active minima
& 278 & 0 & 1 & $+0.003272$\\
\bottomrule
\end{tabular}
\end{table*}

The result is unchanged qualitatively when the nominal step is reduced
by a factor of ten to $\epsilon=10^{-3}$: the active-minima rule
improves 278 of the 279 events and leaves 1 unchanged, with no
observed harmful outcome. Continuation produces 111 improvements and
168 harmful outcomes, while incoming-minimum-only optimization gives
166 improvements and 113 harmful outcomes.

This proves the mechanism test.
The fact that there is a failure of the continuation to be infinite does not mean it can simply be patched up by noting the new active branch and making it optimal. At a bottleneck crossing the minimum gap is locally a lower-envelope problem and the direction must protect all co-active minima. The result is not a general monotonicity theorem: the switching points, active minima and their quantum states are oracle information, and the candidate directions are a finite local menu. Rather, it shows that the mathematical lower-envelope prescription yields the expected local redirection when the active-set information is provided.

\subsection{From the local mechanism to discrete graph search}
\label{sec:discrete_search}

The preceding results concern controlled local directions. We next ask
a more limited practical question: whether this information remains
useful when it is converted into complete unit-edge replacements in a
finite graph search. The exhaustive $N=6$ landscape provides a
numerical reference for this purpose, rather than a scalability claim.

For each of 30 matched objectives and each budget $B=5,\ldots,9$, six
starting supports are used. The guided search ranks legal complete
replacement directions by the local sensitivity, evaluates the finite
gaps of the top eight candidates, and accepts the best improvement.
Table~\ref{tab:production} compares the resulting best-of-six value
with the exhaustive unit-weight reference.

The guided best-of-six median is $R=1$ at every budget, but the lower
tail shows that the discrete landscape remains nontrivial. At $B=5$,
the lower quartile and minimum are 0.917 and 0.530; at $B=6$ they are
0.977 and 0.841. For $B\geq7$, the lower quartile reaches one, while
isolated instances still remain below the exhaustive reference.
Moreover, even the all-neighbor search misses the exhaustive optimum
on some sparse instances. Knowing every improving one-replacement
neighbor is therefore insufficient to guarantee a globally optimal
trajectory.

The sensitivity ranking should also be distinguished from the benefit
of simply evaluating several candidate moves. We therefore compare it
with a random control using the same finite-candidate evaluation cap.
The results are shown in Table~\ref{tab:random}.

\begin{table*}[tp]
\centering
\small
\setlength{\tabcolsep}{4pt}
\caption{
Guided versus evaluation-cap-matched random search on the same
30 instances and six starts. Each of 200 random realizations first
takes a best-of-six ratio; the table averages these ratios or their
hit indicators as appropriate. $\overline R$ is the mean instance ratio and
$\delta\overline R$ is guided minus expected random. Intervals are
pointwise 95\% objective-bootstrap intervals; see Appendix~\ref{app:random}.
}
\label{tab:random}
\begin{tabular}{rrrrrrr}
\toprule
$B$ & Guided hit (\%) & Random hit (\%) &
$\overline R_g$ & $\overline R_r$ &
$\delta\overline R$ & 95\% interval\\
\midrule
5 & 56.7 & 54.5 & 0.9407 & 0.9457 & -0.0050 &
[-0.0344,+0.0203] \\
6 & 70.0 & 66.6 & 0.9781 & 0.9709 & +0.0072 &
[-0.0010,+0.0172] \\
7 & 80.0 & 74.6 & 0.9961 & 0.9866 & +0.0095 &
[+0.0031,+0.0178] \\
8 & 93.3 & 83.0 & 0.9927 & 0.9849 & +0.0078 &
[+0.0018,+0.0158] \\
9 & 93.3 & 86.3 & 0.9980 & 0.9923 & +0.0057 &
[+0.0007,+0.0124] \\
\bottomrule
\end{tabular}
\end{table*}

The guided best-of-six median is $R=1$ at every budget, but the lower
tail shows that the discrete landscape remains nontrivial. At $B=5$,
the lower quartile and minimum are 0.917 and 0.530; at $B=6$ they are
0.977 and 0.841. For $B\geq7$, the lower quartile reaches one, while
isolated instances still remain below the exhaustive reference.
Moreover, even the all-neighbor search misses the exhaustive optimum
on some sparse instances. Knowing every improving one-replacement
neighbor is therefore insufficient to guarantee a globally optimal
trajectory.

\begin{table*}[tp]
\centering
\small
\caption{
Production search at $N=6,K=3$, with 30 matched instances and six
shared starts per instance. Guided search uses $k=8$. Hits satisfy
$R\geq0.999999$ relative to the numerically evaluated exhaustive
reference. Median, $Q_{25}$, and minimum describe the guided
best-of-six $R$. Miss is the percentage of all 180 guided trajectories
that terminate despite the existence of an improving one-edge
neighbor.
}
\label{tab:production}
\begin{tabular}{rrrrrrrr}
\toprule
$B$ & Graphs & Guided hits & All-neighbor hits &
Median & Q25 & Min. & Miss (\%)\\
\midrule
5 & 1,296 & 17/30 & 26/30 & 1.000 & 0.917 & 0.530 & 28.9 \\
6 & 3,660 & 21/30 & 27/30 & 1.000 & 0.977 & 0.841 & 28.9 \\
7 & 5,700 & 24/30 & 26/30 & 1.000 & 1.000 & 0.953 & 23.9 \\
8 & 6,165 & 28/30 & 30/30 & 1.000 & 1.000 & 0.877 & 14.4 \\
9 & 4,945 & 28/30 & 29/30 & 1.000 & 1.000 & 0.947 & 6.1 \\
\bottomrule
\end{tabular}
\end{table*}

The guided best-of-six median is $R=1$ at every budget, but the lower
tail shows that the discrete landscape remains nontrivial. At $B=5$,
the lower quartile and minimum are 0.917 and 0.530; at $B=6$ they are
0.977 and 0.841. For $B\geq7$, the lower quartile reaches one, while
isolated instances still remain below the exhaustive reference.
Moreover, even the all-neighbor search misses the exhaustive optimum
on some sparse instances. Knowing every improving one-replacement
neighbor is therefore insufficient to guarantee a globally optimal
trajectory.

The sensitivity ranking should also be distinguished from the benefit
of simply evaluating several candidate moves. We therefore compare it
with a random control using the same finite-candidate evaluation cap.
The results are shown in Table~\ref{tab:random}.
At $B=5$ and $6$, the guided--random differences are small and their
intervals include zero. At $B=7,8,9$, the mean guided ratios exceed the
random-control means by approximately 0.0095, 0.0078, and 0.0057,
respectively, with the corresponding pointwise intervals above zero.
The discrete-search evidence is therefore deliberately weaker than the
local mechanism result. It shows that the sensitivity can provide
useful candidate ordering in part of the tested finite landscape; it
does not show that local sensitivity defines a globally convergent
graph-search algorithm.

\subsection{How much of the quantum direction can be inferred classically?}
\label{sec:local_selectors}
The oracle results use the low-energy eigenstates at the active
bottleneck. A separate question is whether the favorable local
direction can be inferred from simpler information. We compare four
selectors: the normalized initial-gap spectral direction, the
singleton relax-swap rule, the practical inter-basin coupling score,
and a retrospective oracle-pair version of the same coupling
construction. Their quality is measured by
\[
Q_i(e)=\frac{g_e}{\max_j g_j},
\]
so that $Q\simeq1$ indicates that the selected edge nearly maximizes
the true normalized local sensitivity for that instance.

The results in Table~\ref{tab:selectors} concern 30 accepted objectives
per size--geometry cell, or 270 observations in total. These are
conditioned on one-exchange separation of the two lowest classical
configurations, successful discovery of two basins, a non-endpoint
production bottleneck, and positive maximal raw and normalized
sensitivities; see Appendix~\ref{app:selectors}. They should not be
interpreted as an unfiltered sample of dense objectives. Within this
conditioned ensemble, the results do not support a universal classical
replacement for the quantum sensitivity.

\begin{table*}[tp]
\centering
\small
\setlength{\tabcolsep}{4pt}
\caption{
Refined-point median normalized-sensitivity quality $Q$ on the same
30 accepted objectives per cell, subject to the one-exchange,
two-basin, and sensitivity filters in Appendix~\ref{app:selectors}.
$m_2$ is Fiedler multiplicity and
$n_+$ counts positive spectral initial-gap derivatives over all
physical pairs, with denominator $\binom{N}{2}$. It does not count
positive interior minimum-gap sensitivities. Practical and spectral
selected edges are held fixed; the retrospective oracle pair is
recomputed at the refined minimum. ``Coupling'' is the practical
$\gamma=8$ score. $\simeq0^{-}$ denotes a small negative median.
Rounded values of one do not imply that every case is optimal.
}
\label{tab:selectors}
\ifshowchanges\color{RevisionColor}\fi
\begin{tabular}{@{}rlrrrrrr@{}}
\toprule
$N$ & Geometry & $m_2$ & $n_+$ & Spectral & Singleton & Coupling & Oracle pair\\
\midrule
7 & cycle & 2 & 0/21 & 0.117 & 0.806 & 0.351 & 0.979 \\
7 & star+edge & 4 & 0/21 & 0.003 & 0.829 & 0.872 & 0.996 \\
7 & path+chord & 1 & 17/21 & 0.115 & 0.118 & 0.782 & 1.000 \\
8 & cycle & 2 & 0/28 & 0.084 & 0.203 & 0.578 & 1.000 \\
8 & star+edge & 5 & 0/28 & $\simeq0^{-}$ & 0.602 & 0.422 & 0.521 \\
8 & path+chord & 1 & 24/28 & 0.171 & 0.172 & 0.130 & 1.000 \\
10 & cycle & 2 & 0/45 & 0.000 & 0.241 & 0.115 & 1.000 \\
10 & star+edge & 7 & 0/45 & 0.000 & 0.718 & 0.941 & 1.000 \\
10 & path+chord & 1 & 39/45 & 0.255 & 0.076 & 0.188 & 0.885 \\
\bottomrule
\end{tabular}
\end{table*}

For the cycle and star-plus-edge geometries, Fiedler degeneracy prevents
a positive single-edge initial-gap direction, consistent with the
rank-one limitation derived in Sec.~\ref{sec:limits}. The spectral selector in these
cells therefore identifies the least-unfavorable spectral direction,
not necessarily an improving minimum-gap direction. The
path-plus-chord graph supplies the nondegenerate control. Even there,
the ordering is not universal: the spectral medians are 0.115, 0.171,
and 0.255 for $N=7,8,10$, while the practical coupling medians are
0.782, 0.130, and 0.188. Neither graph connectivity alone nor the
tested problem-informed approximation therefore reproduces the quantum
direction uniformly across sizes and geometries.

The basin construction nevertheless reveals where an important part of
the practical error enters. At the refined bottlenecks, the practical
basin pair agrees with the retrospective crossing-relevant pair in
\PairMatches/\PairTotal{} observations (\PairPct\%). Conditional on a
correct pair, the pooled median practical-coupling quality is
\RightPairQ{}, whereas it falls to \WrongPairQ{} when the selected
pair is incorrect. Replacing only the pair by the retrospective
crossing-relevant pair on those same mismatched observations raises the
median to \WrongPairOracleQ{}, with a median paired gain of
\WrongPairGain{}. Figure~\ref{fig:pair} shows this controlled
diagnostic.

\begin{figure}[htbp]
\centering
\includegraphics[width=\linewidth]{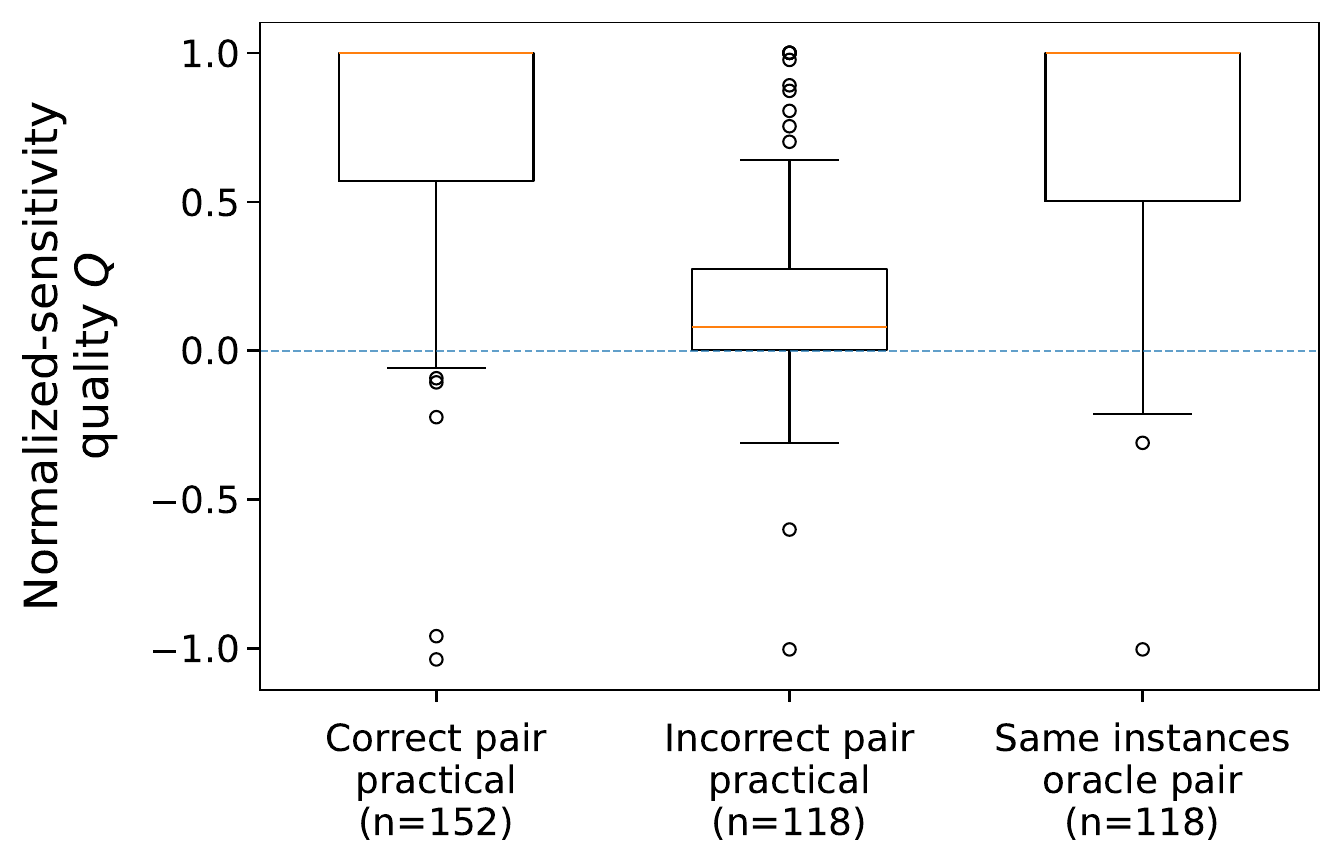}
\caption{
Refined-minimum basin-pair diagnostic pooled descriptively over nine
cells of 30 accepted objectives, using the same conditioned ensemble
as Table~\ref{tab:selectors} and Appendix~\ref{app:selectors}.
There are 152 matching pairs and 118 mismatches. The two
rightmost groups contain the same 118 observations, first using the
practical pair and then the retrospective crossing-relevant pair.
Boxes show medians and interquartile ranges; whiskers use the standard
1.5-IQR convention and points show outliers. Negative $Q$ values are
allowed by Eq.~\eqref{eq:Q}. The retrospective pair is
used only as a diagnostic reference and is not available to the
practical selector.
}
\label{fig:pair}
\end{figure}

Pair discovery is not the only limitation. Even with the retrospective
pair, the median quality is 0.521 in the $N=8$ star-plus-edge cell and
0.885 for $N=10$ path-plus-chord. Unequal mixing, imperfect basin
amplitudes, participation of additional states, and the normalization
term in the exact sensitivity can therefore still matter. The
classical-selector analysis should consequently be read as a
diagnostic of which information is missing, rather than as a scalable
replacement for the quantum-state response.

\section{Discussion}
\label{sec:discussion}

Experiments separate local reaction from finite-step
reliability. All directions chosen by the re-optimization procedure
enhance the re-minimized gap at the two smallest steps in both
families, although the finite update may affect the identity of the
local minimum controlling the global gap. For all 420 of 421
dense switching paths, the first branch exchange happens above
the initial gap. It indicates the operation of the “repair-and-expose” mechanism, in which repairing the active bottleneck exposes a new competing minimum. The sparse family demonstrates the same ``repair-and-
expose'' mechanism but with reduced switching frequency.

For 136 interior branch exchanges, the original direction lifts the
departing branch and drops the entering branch. Associated with it
overlap reduction suggests that migration involves the low-energy
subspace reorganization. Protection of the lower envelope requires accounting for all active minima. The active-minima rule improves 278 of 279 exchanges and abstains at the remaining event, with no harmful outcome, whereas continuation of the original direction and incoming-minimum-only optimization degrade 171 and 115 events, respectively. The smaller step check yields the same conclusion.
These results support relinearization in the neighborhood of a bottleneck exchange, but they do not imply improvement under an arbitrary finite step.

The discrete search procedure verifies whether this quantum behavior retains its significance as an ordering of full replacements. It often does at the
tested $N=6$ scale; the evaluation-cap-matched comparison shows a
modest advantage at larger budgets. Even all-neighbor search can
stop below the exhaustive optimum, demonstrating the limits of
local information. Likewise, neither graph-spectral nor the tested
classical selectors dominate uniformly. The basin-pair experiment
identifies one source of error: a useful coupling score depends on
identifying the competitors that actually define the bottleneck.
Finite driver changes make this identity dynamical.

The results use small feasible sectors, exact objective normalization,
exhaustive graph landscapes, and oracle low-energy states as benchmark
resources. They establish a finite-size mechanism rather than a
scalable prescription or finite-time quantum advantage. The sparse
family tests sensitivity and migration, but not the full search or
selector pipeline. Relinearization uses resolved interior switching
points and their quantum states; its 279 events arise from 27
objectives and are repeated observations. Selector results concern
an explicitly conditioned sample, and continuous minima are found
numerically.

A scalable method would additionally need to detect competing
bottlenecks, estimate their active-set responses from accessible
information, and choose steps that retain local validity. The present
analysis specifies the response and lower-envelope structure such
an extension must respect.

\section{Conclusion}
\label{sec:conclusion}

We have formulated sparse constraint-preserving driver design as a
local minimum-gap response problem. For fixed-cardinality Laplacian
drivers, the response can be resolved into physical exchange channels
through the low-energy Dirichlet-energy contrast, including the
correction arising from driver normalization. Numerical tests in two
objective families confirm that this sensitivity reliably identifies
locally favorable driver changes.

The central finite-step effect is bottleneck migration. A direction
that repairs the current minimum gap can expose a different competing
minimum, so that a locally correct direction ceases to optimize the
global minimum gap. At resolved exchanges of the active bottleneck,
relinearizing the response over all co-active minima is substantially
more reliable than continuing the original direction or optimizing
only the incoming minimum.

Thus the relevant object for driver design is not a universally
preferred graph topology, but the evolving low-energy bottleneck of
the annealing path. The present results establish this local
sensitivity--migration--relinearization mechanism; they do not
establish global convergence, scalable autonomous driver design,
finite-time speedup, or quantum advantage.

\appendix
\section{Numerical methods and reproducibility}
\label{app:repro}

The calculations employ the normalised Hamiltonian of Eq.~\eqref{eq:annealing_path}. Gap minima and exchange points are numerical estimates. The implementation, seeds and solver settings are documented in the accompanying code and numerical specification in detail.

\subsection{Instances and numerical conventions}
\label{app:generators}
\label{app:locators}

We lexicographically order feasible $K$-subsets and physical pairs $(i,j)$, $i<j$, and define the support mask $m(S)=\sum_{e\in S}2^{\operatorname{index}(e)}$ by the edge indices.
NumPy \texttt{default\_rng} generators are advanced by rejected draws.
In the dense generator in Sec.~\ref{sec:protocol}, dense objective $i$ uses seed $1000003N+9176i$ and a $N\times N$ matrix $M$ with standard normal entries.
Seed for the frustrated family: \texttt{SeedSequence([20260928,31001,i])}. Edges for a randomly relabelled triangular prism. Before relabelling: \{01,12,02,34,45,35,03,14,25\} where $ij$ means the unordered pair $\{i,j\}$.
Both bond signs and frustration are needed. The two lowest viable costs in each family must differ by more than $10^{-8}$. Rejected complete draws are retried, up to 10,000 times.
The dense and sparse comparisons use $i=0,\ldots,29$.

Starting supports use seed $20260922+10000B+i$ and uniform draws of $B$ distinct edges, rejecting disconnected graphs. The first six accepted supports define the starts, where duplicates across starts are allowed. Exhaustive enumeration considers all connected labelled support, without identifying isomorphic graphs.

The production gap locator samples 41 equally spaced points in $[0,1]$ and does three 13-point zooms, with the initial half-width $1/40$ reduced by six after each zoom. The retained point is replaced by strictly smaller gaps. The dense oracle locator instead samples 161 uniform points with 49 logarithmic endpoint distances from $10^{-8}$ to $10^{-2}$ and refines sampled valleys with bounded minimisation at position tolerance $10^{-13}$. The frustrated-family oracle locator is a 321 point profile.
Eigenvalue clusters use tolerance $10^{-10}+10^{-9}\max(1,|\lambda_*|)$ about $\lambda_*$. Directional responses use the appropriate extremal eigenvalue in complete degenerate eigenspaces. An unresolved ground-state cluster results in numerical abstention.

The fixed oracle directions in Figs.~\ref{fig:oracle_steps}--\ref{fig:subspaces} and the corresponding Table~\ref{tab:family_comparison} calculations achieve scores within $10^{-12}$ of the maximum are treated as tied, and the candidate with the smallest proposed support mask is selected. An initial $s^\star\geq1-10^{-10}$ to cause endpoint abstention; a selected score at most zero to cause score abstention. The refined locator and a threshold of $10^{-10}$ give fractional improvements, and the production locator and a threshold of $10^{-12}$ give complete improvements.

\subsection{Continuation, isolation, and relinearization}
\label{app:continuation}
\label{app:isolation}
\label{app:relinearization}

Continuation uses 641 uniform s points, 65 logarithmic endpoint distances in the range $10^{-10}$ to $10^{-2}$ and bounded and derivative-root refinement of observed valleys. The first update grid has 81 uniform $\alpha$ points plus $0.001$ and $0.01$. Then, interior branches are matched one-to-one by minimising total $s$ displacement, accepting matches within $0.15$; endpoints retain their identity. For active-minimum motion larger than $0.025$ or a coexisting, nearly tied, or unresolved transition, intervals are bisected down to width $10^{-4}$ with a 450-profile limit.

For a resolved exchange, there need to be coexisting local
minima that differ by at least $0.02$ in $s$, opposite gap orderings,
and equality between them refined in the interval. Gap-depth comparisons use
$10^{-10}+10^{-8}\max(|a|,|b|)$; the refined equality and global-minimum
errors must each be below $10^{-9}$.
Endpoint touches and unresolved switches are stored separately. ``No
resolved switch'' implies there were none in accordance with this
protocol. For the dense run, there are neither unresolved nor capped
walks, for the sparse run, there is one unresolved walk and no capped
walks. There are no invalid rank two overlap comparisons in either set.

Rank-two overlap condition calls for $E_2-E_1>10^{-9}$
in the starting and ending active minima. Failed cases are marked and
discarded, and not arbitrary eigenvectors from some cluster chosen.
All of 1,596 dense terminal evaluations meet the test, among which there
are 19 starting minima at $s=0$; minimum gaps are $0.00352365$
initially and $0.00196139$ finally.

Of 798 dense starts chosen, 797 of them have an improved complete
neighbor, and there are 662 improvements for chosen replacements. There
is one special start of $B=5$, $i=11$, $r=0$, whose masks are
$26884\to24838$. Its score is about $0.158432$, and its gap gets an
increase by $0.00157281$ when $\alpha=0.01$, but there is no improved
production gap due to its complete neighbors. Its chosen full replacement
reduces the gap from $0.0693910$ to $0.0325896$. Also,
all 102 endpoint-excluded starts do not have any improved complete
neighbors.

At an interior exchange, all the retained minimum points within $10^{-10}$
of the global gap are part of the active set. The weight shift menu transfers weight from each positive edge to another edge according to removal and addition indices. In the incoming and active minima criterion, the first maximizing item is selected, with an abstention if the best criterion value is at most $10^{-10}$. For all the criteria, $t=\min(\epsilon, w_{\min}/4)$, where $\epsilon=10^{-2}$ in Table~\ref{tab:adaptive_primary} and $10^{-3}$ in the finer step size test. Re-minimization also explores around the active positions. Changes greater than $10^{-10}$ improve things, less than $-10^{-10}$ hurt and otherwise leave things unchanged.

\subsection{Search controls, aggregation, and selector conditioning}
\label{app:random}
\label{app:selectors}
Guided production search sorts all legal connected one-edge-replacement neighbors in
order of their scores using the edge tuple in the lexicographic
way to break ties. It examines the top eight and accepts the
candidate with the largest increase in production gap greater than
$10^{-12}$. The all-neighbor comparator examines all neighbors.
Both terminate upon no progress or after 50 comparisons. Every
query made by the candidate is counted, except for the initial
gaps, sensitivity evaluations, and offline exhaustive evaluation.

A random restart uses the same cap for the matched guided
restart's evaluation. It randomly shuffles the legal neighbors and
tests in batches of at most eight until the available cap is used
up. It accepts the first member of the batch that makes
improvement and terminates when the cap runs out or no batch
succeeds. This matches candidate
caps, not total runtime.
For objective $i$, realization $m$ first computes
\[
R^{\rm rand}_{i,m}=
\frac{\max_{r=0,\ldots,5}J^{\rm rand}_{i,m,r}}{J_i^\star},
\qquad
\overline R^{\rm rand}_i=\frac1{200}\sum_{m=1}^{200}R^{\rm rand}_{i,m}.
\]
Table~\ref{tab:random} averages these values across 30 objectives;
random hit rates average the individual indicators
$\mathbf1\{R^{\rm rand}_{i,m}\geq0.999999\}$.
The 200 realizations comprise 100 for each base seed $20260922$ and
$20270001$, using
\texttt{SeedSequence([base\_seed,realization,B,i,r])}.
The paired guided-minus-random intervals use 20,000 objective
bootstrap samples and linear 2.5th/97.5th percentiles. The bootstrap RNG uses seed $20260922$ and is advanced in budget order. Random means remain fixed during
resampling; intervals are pointwise and conditional on those means.

For $N=7,8,10$, $K=\lfloor N/2\rfloor$, $B=N$, supports are
\[
\begin{split}
S_{\rm cycle}&=\{\{i,(i+1)\bmod N\}:0\leq i<N\},\\
S_{\rm star+edge}&=\{\{0,i\}:1\leq i<N\}\cup\{\{1,2\}\},\\
S_{\rm path+chord}&=\{\{i,i+1\}:0\leq i<N-1\}\cup\{\{0,2\}\}.
\end{split}
\]
Geometry codes $c=11,23,37$ give dense proposal index
$700000+10000N+100c+j$. The first 30 accepted proposals per cell
require a one-exchange separation of the two lowest classical
configurations, a non-endpoint production bottleneck, positive
maximal raw and normalized scores above $10^{-12}$, and discovery
of two distinct basins.

Projected-gradient relaxation is rounded to the $K$ largest
coordinates. Basin descent uses supported improving exchanges;
pair discovery starts from the rounded state and 20 uniformly
sampled feasible states, using seed $900000+1000N+100(j+1)$.
The coupling parameter is $\gamma=8$. Practical and spectral edges
are selected before refinement and held fixed; selector ties use
the first maximizing physical-pair entry. Refinement utilizes 641 equal points, 61 logarithmic endpoints, and probes around
the production minimum. The retrospective pair employs the pair of basins that have
the largest weight sum for both ground and first excited states at the refined
minimum. An accepted pair having maximum sensitivity at most $10^{-12}$ generates an
error but not an outcome-dependent redraw. Figure~\ref{fig:pair} is a comparison of
the 118 mismatching observations before and after replacing the pair.
\section*{Acknowledgements}
Aruna Gupta and S. R. Hassan acknowledge support from the QC BRIDGE Post-Doctoral Fellowship, funded by the Foundation for QC Innovation (FQCI), DST-NQM Hub for Quantum Computing at the Indian Institute of Science, Bengaluru.
\bibliography{references}
\end{document}